%% file: main.tex
\documentclass[11pt,a4paper]{article}
\usepackage[lmargin=1.0in,rmargin=1.0in,bottom=1.0in,top=1.0in,twoside=False]{geometry}

\usepackage{authblk}

\usepackage{amssymb,amsmath}
\usepackage{graphicx}%
\usepackage{enumerate}
\usepackage{tikz}
\usetikzlibrary{shapes}

\usepackage[T1]{fontenc}

 \usepackage{xcolor}
 \usepackage{mathtools}
\usepackage{microtype}
\usepackage{amsfonts}
\usepackage{comment}
\usepackage{mathrsfs}
\usepackage[vlined, ruled, linesnumbered]{algorithm2e}
\DontPrintSemicolon

\usepackage{trimspaces}
\usepackage{nccfoots}
\usepackage{setspace}
\usepackage{inconsolata}
\usepackage{libertine}
\usepackage[absolute]{textpos}
\usepackage{thmtools}
\usepackage{thm-restate}

\usepackage{todonotes}

\usepackage{longtable}

\definecolor{blue}{rgb}{0.1,0.2,0.5}
\definecolor{brown}{rgb}{0.6,0.6,0.2}
\usepackage[ocgcolorlinks, linkcolor={blue}, citecolor={brown}]{hyperref}

\usepackage{comment}

\usepackage[amsmath,thmmarks,hyperref]{ntheorem}
\usepackage{cleveref}

\crefformat{page}{#2page~#1#3}%
\Crefformat{page}{#2Page~#1#3}%
\crefformat{equation}{#2(#1)#3}%
\Crefformat{equation}{#2(#1)#3}%
\crefformat{figure}{#2Figure~#1#3}%
\Crefformat{figure}{#2Figure~#1#3}%
\crefformat{section}{#2Section~#1#3}
\Crefformat{section}{#2Section~#1#3}
\crefformat{chapter}{#2Chapter~#1#3}
\Crefformat{chapter}{#2Chapter~#1#3}
\crefformat{chapter*}{#2Chapter~#1#3}
\Crefformat{chapter*}{#2Chapter~#1#3}
\crefformat{part}{#2Part~#1#3}
\Crefformat{part}{#2Part~#1#3}
\crefformat{enumi}{#2(#1)#3}
\Crefformat{enumi}{#2(#1)#3}

\usepackage{enumerate}

\usepackage{latexsym}

\theoremnumbering{arabic}
\theoremstyle{plain}
\theoremsymbol{}
\theorembodyfont{\itshape}
\theoremheaderfont{\normalfont\bfseries}
\theoremseparator{.}

\newtheorem{theorem}{Theorem}
\crefformat{theorem}{#2Theorem~#1#3}
\Crefformat{theorem}{#2Theorem~#1#3}

\newcommand{\newtheoremwithcrefformat}[2]{%
  \newtheorem{#1}[theorem]{#2}%
  \crefformat{#1}{##2\MakeUppercase#1~##1##3}%
  \Crefformat{#1}{##2\MakeUppercase#1~##1##3}%
}
\newcommand{\newseptheoremwithcrefformat}[2]{%
  \newtheorem{#1}{#2}%
  \crefformat{#1}{##2\MakeUppercase#1~##1##3}%
  \Crefformat{#1}{##2\MakeUppercase#1~##1##3}%
}

\newtheorem{lemma}[theorem]{Lemma}
\newtheoremwithcrefformat{proposition}{Proposition}
\newtheoremwithcrefformat{observation}{Observation}
\newtheoremwithcrefformat{conjecture}{Conjecture}
\newtheoremwithcrefformat{corollary}{Corollary}
\newseptheoremwithcrefformat{claim}{Claim}
\theorembodyfont{\upshape}
\newtheoremwithcrefformat{example}{Example}
\newtheoremwithcrefformat{remark}{Remark}
\newseptheoremwithcrefformat{definition}{Definition}
\newseptheoremwithcrefformat{question}{Question}

\crefname{theorem}{Theorem}{Theorems}

\theoremstyle{nonumberplain}
\theoremheaderfont{\scshape}
\theorembodyfont{\normalfont}
\theoremsymbol{\ensuremath{\square}}
\newtheorem{proof}{Proof}

\theoremsymbol{\ensuremath{\lrcorner}}
\newtheorem{claimproof}{Proof of Claim}

\def\cqedsymbol{\ifmmode$\lrcorner$\else{\unskip\nobreak\hfil
\penalty50\hskip1em\null\nobreak\hfil$\lrcorner$
\parfillskip=0pt\finalhyphendemerits=0\endgraf}\fi}

\tikzset{
    position/.style args={#1:#2 from #3}{
        at=(#3.#1), anchor=#1+180, shift=(#1:#2)
    }
}

\newcommand{\Oh}{\mathcal{O}}

\newcommand{\ceil}[1]{\left \lceil #1 \right \rceil }

\let\originalleft\left
\let\originalright\right
\renewcommand{\left}{\mathopen{}\mathclose\bgroup\originalleft}
\renewcommand{\right}{\aftergroup\egroup\originalright}

\renewcommand{\leq}{\leqslant}
\renewcommand{\geq}{\geqslant}

\newcommand{\defproblem}[3]{
  \vspace{3mm}
  \noindent\fbox{
  \begin{minipage}{0.97\textwidth}  
  #1 \\ 
  {\bf{Input:}} #2  \\
  {\bf{Output:}} #3
  \end{minipage}
  }
  \vspace{3mm}
}

\usepackage{todonotes}

\definecolor{lars}{HTML}{0F3CCB}

\definecolor{larsA}{HTML}{196F3D}
\definecolor{larsB}{HTML}{FF8000}
\definecolor{larsC}{HTML}{DF0101}

\usepackage{lineno}

\usepackage{macros}

\newcommand{\frontpageformat}{arxiv}

\begin{document}

\ifthenelse{\equal{\frontpageformat}{submission}}{%
\author{anonymous}
\title{A tight quasi-polynomial bound for Global Label Min-Cut}
\begin{titlepage}
\def\thepage{}
\thispagestyle{empty}
\maketitle
}{%
\author[1,2]{Lars Jaffke}
\author[3]{Paloma T.\ Lima}
\author[1]{Tom\'{a}\v{s} Masa\v{r}\'{i}k}
\author[1]{Marcin Pilipczuk}
\author[1,4]{Ueverton S.\ Souza}

\affil[1]{University of Warsaw, Poland}
\affil[2]{University of Bergen, Norway}
\affil[3]{IT University of Copenhagen, Denmark}
\affil[4]{Fluminense Federal University, Brazil}

\title{A tight quasi-polynomial bound for Global Label Min-Cut%
\thanks{This research is a part of a project that has received funding from the European Research Council (ERC)
under the European Union's Horizon 2020 research and innovation programme
Grant Agreement 714704 (LJ, TM, MP, US) and from the Research Council of Norway (LJ).}}

\begin{titlepage}
\def\thepage{}
\thispagestyle{empty}
\maketitle
\begin{textblock}{20}(0, 13.3)
\includegraphics[width=40px]{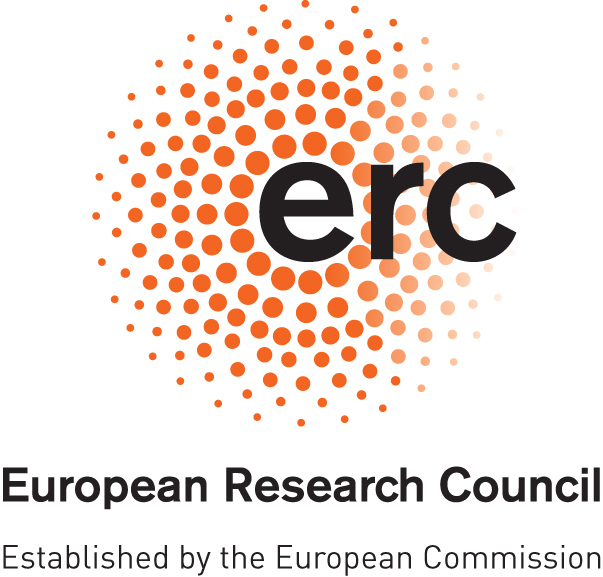}%
\end{textblock}
\begin{textblock}{20}(0, 14.1)
\includegraphics[width=40px]{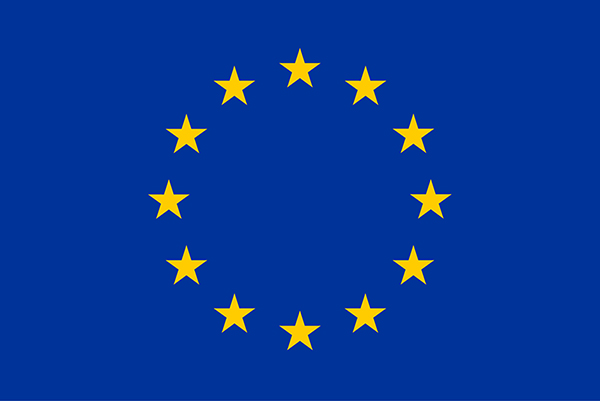}%
\end{textblock}
}

\begin{abstract}
We study a generalization of the classic \textsc{Global Min-Cut} problem, called \textsc{Global Label Min-Cut} (or sometimes \textsc{Global Hedge Min-Cut}):
the edges of the input (multi)graph are labeled (or partitioned into color classes or hedges), and removing all edges of the same label (color or from the same hedge) costs one.
The problem asks to disconnect the graph at minimum cost.

While the $st$-cut version of the problem is known to be \NP-hard, the above global cut version is known to admit a quasi-polynomial randomized $n^{\Oh(\log \mathrm{OPT})}$-time algorithm
due to Ghaffari, Karger, and Panigrahi [SODA 2017].
They consider this as ``strong evidence that this problem is in \Poly''. We show that this is actually not the case.
We complete the study of the complexity of the \textsc{Global Label Min-Cut} 
problem by showing that the quasi-polynomial running time is probably optimal:
We show that the existence of an algorithm with running time $(np)^{o(\log n/ (\log \log n)^2)}$ would contradict the Exponential Time Hypothesis,
  where $n$ is the number of vertices, and $p$ is the number of labels in the input.
The key step for the lower bound is a proof that \textsc{Global Label Min-Cut} is \W[1]-hard when parameterized by the \emph{number of uncut labels}.
In other words, the problem is difficult in the regime where almost all labels need to be cut to disconnect the graph. 
To turn this lower bound into a quasi-polynomial-time lower bound, we also needed to revisit the framework due to Marx [Theory Comput. 2010] of proving lower bounds assuming Exponential Time Hypothesis
through the \textsc{Subgraph Isomorphism} problem parameterized by the number of edges of the pattern.
Here, we provide an alternative simplified proof of the hardness of this problem that is more versatile with respect to the choice of the regimes of the parameters.
\end{abstract}

\end{titlepage}

\input{intro}

\input{sec-hardness}

\input{cybt}

\input{wrap-up}

\bibliographystyle{plainurl}
\bibliography{references}
\end{document}

%% file: intro.tex
\section{Introduction}

Given a weighted graph $G=(V,E)$ with $n$ vertices and $m$ edges,  in the (Global) {\sc Min-Cut} problem we are asked to find a cut $[S,T]$ of $G$ minimizing the total weight of edges crossing from $S$ to $T$,  where a cut of $G$ is a partition of $V$ into two non-empty subsets $S$ and $T=V\setminus S$.  This problem is one of the most fundamental problems in computer science, and dates back to the earliest days of the field~\cite{ford_fulkerson_1956}.

An $st$-cut of a graph $G$ is a cut $[S,T]$ having $s\in S$ and $t\in T$, and a minimum cut of $G$ is a cut having an edge-cut set of minimum total weight.
We denote by {\sc Min-$st$-Cut} the problem of computing the minimum $st$-cut of a given graph.
A minimum $st$-cut of a graph $G$ can be found by using any algorithm for computing a maximum $st$-flow.
While the history of minimum $st$-cut and maximum flow algorithms dates back to the
1950s~\cite{ford_fulkerson_1956,ford1962dr}, revolutionary results are obtained regularly up to today:
we mention~\cite{FOCSYu,SiCompTarun,MadryFOCS}, and, most notably, the recent breakthrough near-linear algorithm for polynomially
bounded capacities~\cite{near-lin-flow}.

Due the duality of maximum $st$-flows and minimum $st$-cuts,  a typical strategy to solve the {\sc Min-Cut} problem consists of finding minimum $st$-cuts for a fixed vertex $s$ and all $n-1$ possible choices of $t$ and then selecting the smallest one.  However,  a few number of non-flow-based methods for solving {\sc Min-Cut} 
have been developed~\cite{icalp2020,KargerAlgorithm,Karger2000,matula1993linear,SIDMANagamochi,StoerAlgorithm}.
In particular,  in 1993, Karger~\cite{KargerAlgorithm} developed a simple randomized algorithm to compute a minimum cut of a connected graph, 
   which later inspired many subsequent works.

A natural generalization of {\sc Min-Cut} on graphs is the {\sc Min-Cut} problem on hypergraphs, which aims to determine the smallest number of hyperedges to be removed to disconnect a given hypergraph. %
In 1997, Stoer and Wagner~\cite{StoerAlgorithm} presented a non-flow-based algorithm for {\sc Min-Cut} on hypergraphs with $O(n\cdot m+n^2\log n)$ running time.
This approach was generalized by Queyranne~\cite{QueyranneSubmodular,queyranne1998minimizing} for minimizing symmetric submodular functions.  
In~\cite{QueyranneSubmodular,queyranne1998minimizing}, it was shown that if $V$ is a finite set and $f : 2^V\rightarrow \mathbb{R}$ is a polynomial-time computable function that is symmetric and submodular then there is a polynomial-time algorithm that minimizes $f$ over all proper subsets of $V$. 
	Rizzi~\cite{Rizzi00} observed that the hypergraph cut function is symmetric and submodular,
	which implies that Min-Cut on hypergraphs can be solved by using polynomial-time algorithms
	for submodular function minimization~\cite{QueyranneSubmodular,queyranne1998minimizing}.
Note that {\sc Min-Cut} on hypergraphs is a particular case of a more general cut problem on edge-colored multigraphs: if the vertex set of each hyperedge $i$ is viewed
as a connected subgraph whose edges are colored with color $c_i$ then the goal is to find a cut $[S,T]$ minimizing the number of colors in $\partial(S)$, the set of edges with exactly one endpoint in $S$.

In this paper, we consider the more general problem where each color class is allowed to induce more than one connected component. 
Let $G=(V,E)$ be an edge-colored multigraph with $n$ vertices, $m$ edges,  and an edge coloring $c:E\to \{1,2,\ldots,p\}$,  not necessarily proper. 
We let 
$c(\partial(S))$ be the set of colors that appear in $\partial(S)$. 
\CMC is the problem of finding a subset $S \subseteq V$, $S \neq \emptyset, V$ 
that minimizes the size of $c(\partial(S))$.

\defproblem
	{\CMC}
	{A multigraph $G=(V,E)$ with an edge coloring $c:E \to \{1,2,\ldots,p\}$ and an integer $k>0$.}
	{\yes if $G$ has a proper subset $\emptyset \neq S\subsetneq V$ such that  $|c(\partial(S))|\leq k$, \no otherwise.}
	
\CMC naturally captures a problem in network survivability
in which the failure of a single connection implies failure of a whole set of connections, see for instance~\cite{coudert2016combinatorial,CoudertDPRV07,GhaffariKP17}.
In the literature, the colors are sometimes called \emph{labels} or \emph{hedges}; in this work, we stick to \emph{colors}.

Similarly, in \CMCst, we are asked to find an edge-cut set that separates a given pair $s,t$ of vertices using at most $k$ colors (or,  as few colors as possible).
Note that the analogously defined \textsc{Min-Cut} on ``hyperedge-colored'' hypergraphs is in fact the same problem as \CMC{}.

The computational complexity of \CMCst and \CMC has been widely investigated in recent years~\cite{Blin201466,ItorProtti,CoudertDPRV07,coudert2016combinatorial,fellows2010parameterized,GhaffariKP17,Linqing,zhang2014,zhang2011approximation,zhang2016label,Zhang2018,ZHANG2020INFCOMP}.
As first observed in~\cite{zhang2014},  if we define a function $f$ on $2^V$ that counts the number of colors in an edge-cut set resulting from a vertex subset $S$, that is, $f(S) = c(\partial(S))$, then we can easily verify that $f$ is not submodular.
Therefore we cannot use the same approach as for \textsc{Min-Cut} on hypergraphs via submodular function minimization to solve \CMC.

In the special case when every color contains at most $r$ edges, 
Coudert et al.~\cite{CoudertDPRV07} showed that \CMC can be solved in polynomial time, but  \CMCst is \NP-hard even when each color contains at most two edges.  
Furthermore, Blin et al.~\cite{Blin201466} presented a randomized 
polynomial-time
algorithm for \CMC that returns an optimal colored cut of $G$ with probability at least~$|V|^{-2r}$.

Several approximation and hardness results for  \CMCst are presented in~\cite{Linqing,zhang2014,zhang2011approximation,Zhang2018,ZHANG2020INFCOMP}.  
Zhang and Fu~\cite{zhang2016label} showed that \CMCst is \NP-hard even if the maximum length of any path is equal to two.  %
Regarding the parameterized complexity of \CMCst, Fellows et al.~\cite{fellows2010parameterized} showed that  the problem is \W[2]-hard when parameterized by the number of colors of the solution, and  \W[1]-hard when parameterized by the number of edges of the solution.   
Coudert et al.~\cite{coudert2016combinatorial} showed that \CMC can be solved in time $2^s\cdot n^{O(1)}$, where $s$ is the number of colors that induce more than one nontrivial connected component.

Most importantly, Ghaffari, Karger, and Panigrahi~\cite{GhaffariKP17} showed that \CMC can
be solved exactly in time $n^{\Oh(\log \mathrm{OPT})}$, where $n$ is the number of vertices of the graph
and $\mathrm{OPT}$ denotes the cost of the optimum solution,
and a $(1+\varepsilon)$-approximate solution can be found in time $n^{\Oh(\log(1/\varepsilon))}$. 
This results shows that \CMC is probably not \NP-hard, contrary to its $st$-cut variant, \CMCst. 

Despite the above progress, 
the exact computational complexity of \CMC was left open by all these works. 
To the best of our knowledge, the first time that the computational complexity of \CMC was mentioned as an open question it was in~\cite{CoudertDPRV07}.
Note that the result of~\cite{GhaffariKP17} places \CMC between the classes of problems solvable in polynomial and quasi-polynomial time.
After 15 years of research, the community is still asking for an answer to this question.
More recently, in~\cite{ZHANG2020INFCOMP}, Zhang and Tangi explicitly stated that ``a challenging problem is to determine the exact complexity of {\sc Global Label Cut}.''
In addition, prior to~\cite{GhaffariKP17} it was pointed out~\cite{CoudertDPRV07,zhang2014} that
unlike the situation of {\sc Min-Cut} and {\sc Min $st$-Cut} where both are polynomial-time solvable,  
\CMC seems to behave differently than its local counterpart  (\CMCst) and should be easier than it;
the result of~\cite{GhaffariKP17} confirms this suspicion.

\paragraph{Our results.}
Ghaffari, Karger, and Panigrahi write that their randomized $n^{\Oh(\log \mathrm{OPT})}$-time algorithm provides ``strong evidence that this problem is in \Poly''~\cite[abstract]{GhaffariKP17}.
We show that this is probably not the case, and settle the question of the complexity of \CMC:
one of our main results states that the quasi-polynomial complexity is tight under the Exponential Time Hypothesis~\cite{eth}.%
\footnote{The Exponential Time Hypothesis, a now-standard assumption in fine-grained complexity, 
    together with the Sparsification Lemma asserts that one cannot
resolve satisfiability of an $n$-variable $m$-clause 3-CNF formula in time $2^{o(n+m)}$~\cite{eth}.}
\begin{restatable}{theorem}{thmhard}
\label{thm:hard}
Unless the Exponential Time Hypothesis fails, 
there is no algorithm solving \CMC in time
$(np)^{o(\log n/ (\log \log n)^2)}$. 
\end{restatable}

Observe that the lower bound of Theorem~\ref{thm:hard} is tight up to a $\Oh((\log \log n)^2)$ factor in the exponent
with the algorithm of~\cite{GhaffariKP17}, as $\mathrm{OPT} \leq p$ and $n^{\Oh(\log \mathrm{OPT})} = \mathrm{OPT}^{\Oh(\log n)} \leq p^{\Oh(\log n)}$. 
The key ingredient in the proof of Theorem~\ref{thm:hard} is the proof that \CMC is \W[1]-hard
when parameterized by \emph{the number of uncut colors}, that is, the value $p-k$.
\begin{restatable}{theorem}{thmwhard}
\label{thm:w1hard}
The \CMC problem, parameterized by $a := p-k$, the number of uncut colors, is \W[1]-hard.
Furthermore, unless the Exponential Time Hypothesis fails,
  there is no computable function $f$ and an algorithm for the problem with running time bound $f(a) \cdot (np)^{o(a/\log a)}$. 
\end{restatable}

Theorem~\ref{thm:w1hard} is proven in Section~\ref{sec:hardness}.
A natural step to obtain Theorem~\ref{thm:hard} from Theorem~\ref{thm:w1hard} is to pipeline the reduction
of Theorem~\ref{thm:w1hard} (if it started in the classic \textsc{Multicolored Clique} problem) with a standard reduction
from \textsc{3-CNF SAT} to \textsc{Multicolored Clique} that makes a correct choice of the parameters
(here, from an input $n$-variable $m$-clause formula we would want a \textsc{Multicolored Clique} instance
 with roughly $2^{\Oh(\sqrt{n+m})}$ vertices and looking for a clique of size $\Oh(\sqrt{n+m})$). 
Unfortunately, the reduction of Theorem~\ref{thm:w1hard} makes heavy use of the so-called edge-choice gadgets
and therefore starts instead 
in the \textsc{Subgraph Isomorphism} problem, parameterized by the number of edges of the pattern.
Hardness of this source problem is provided by Marx~\cite{Marx2010}; his reduction is involved
and not directly amenable to such choose-convenient-range-of-parameters tricks.

To cope with this issue, we revisit the framework of Marx~\cite{Marx2010} and provide a simplified and streamlined proof
of the hardness of \textsc{Subgraph Isomorphism} that is now versatile to a choice of parameter range.
The exact statement (and necessary definitions) can be found in Section~\ref{sec:fine:Marx}. 
(We remark that the main goal of the framework of Marx~\cite{Marx2010} was to provide hardness for the \emph{treewidth parameterization}
and the ``number of edges'' parameter was only a side corollary there; 
our simplified proof is tailored for the latter parameter and has no implications for the treewidth parameterization.)

%% file: sec-hardness.tex
\newcommand\PDualCMC{\textsc{Partitioned Dual Colored Min-Cut}\xspace}
\newcommand\PSIshort{PSI\xspace}
\newcommand\DualCMCshort{DCMC\xspace}
\newcommand\PDualCMCshort{PDCMC\xspace}
\section{Hardness of \DualCMC}\label{sec:hardness}
In this section we give the main reduction that leads to 
the quasi-polynomial time lower bound under the \ETH for \CMC,
and a proof of \cref{thm:w1hard}.
To give the reduction, we slightly change the perspective on the problem.
A computationally expensive case in the algorithm of~\cite{GhaffariKP17} is 
when a large number of colors is incident with a large number of vertices. %
In particular, this forces a high depth of recursion in the execution of the algorithm. 
In this situation, each color is incident with many vertices, so only very few (say at most $a$) colors
do not appear in an optimum solution. 
The algorithm of~\cite{GhaffariKP17} in fact in this case
reverts to brute-forcing all $\binom{p}{a}$ sets of colors that 
are not in the solution (recall that $p$ is the total number of colors).

This scenario reveals that the number of colors \emph{not} present in a minimum colored cut 
is a relevant parameter for understanding the computational complexity of \CMC.
We therefore consider its dual problem, parameterized by solution size.
(Note that finding the maximum-size set of colors that do not belong to some cut of a graph $G$ is equivalent to finding the maximum number of colors such that the union of their edges forms a disconnected subgraph of $G$.)

\defproblem
	{\DualCMC}
	{Vertex set $W$, graphs $G_1, \ldots, G_p$ on vertex set $W$, integer $a$.}
	{\yes if there exists a set $I \subseteq [p]$ with $\card{I} = a$ 
		such that $(W, \cup_{i \in I} E(G_i))$ is disconnected, \no otherwise.}

We will confirm our intuition that indeed, the case described above makes the problem computationally hard, 
by giving a reduction that proves \W[1]-hardness of \DualCMC parameterized by $a$,
and, together with a reduction given in \cref{sec:wrap}, an
$(np)^{o(\log n/(\log\log n)^2)}$ time lower bound for \CMC under the \ETH.
The reduction is from the \PSI problem.

\defproblem
	{\PSI (\textsc{PSI})}
	{Graph $H$, called \emph{pattern}, graph $K$, called \emph{host}, with vertex partition $\uplus\{V_x \mid x \in V(H)\}$
		such that $\card{V_x} = n$ for all $x \in V(H)$.}
	{\yes if there is a set $\{v_x \in V_x \mid x \in V(H)\} \subseteq V(K)$ such that 
		for all $x, y \in V(H)$, $xy \in E(H)$ implies $v_x v_y \in E(K)$, \no otherwise.}
		
To give an overview of the main idea behind the reduction of the following theorem,
we consider for one moment a \emph{partitioned} version of \DualCMC, 
where in addition we are given a partition $(\calG_1, \ldots, \calG_a)$ of $\{G_1, \ldots, G_p\}$
and we require a solution to select one graph from each $\calG_i$.

The following simple reduction from \PSIshort to \PDualCMC (\PDualCMCshort) 
would suffice if we could restrict our attention to this problem
instead of the non-partitioned version.
Let $(H, K)$ be an instance of \PSIshort where $H$ is the pattern graph and $K$ the host graph,
and let $E(H) = \{e_1, \ldots, e_a\}$.
We may assume that $H$ is connected.
We construct an instance $(W, \calG_1, \ldots, \calG_a)$ of \PDualCMCshort as follows.
We let $W = \{t\} \cup V(K)$.
For all $i \in [a]$, let $e_i = xy$, 
and consider the set of edges $E(V_x, V_y) \subseteq E(K)$ with one endpoint in $V_x$ and the other in $V_y$.
For all $v_xv_y \in E(V_x, V_y)$, we add a graph $G(i, v_x, v_y)$ to $\calG_i$ that consists of the edge $v_xv_y$, 
and all edges $tz$, where $z \in V_x \cup V_y \setminus \{v_x, v_y\}$ (cf.~\cref{fig:red:1}).

For one direction, suppose that $X = \{v_x \in V_x \mid x \in V(H)\}$ is a solution to $(H, K)$.
Then, taking the union over all $e_i = xy$ of the graphs $G(i, v_x, v_y)$ (where $v_x, v_y \in X$)
is a solution to $(W, \calG_1, \ldots, \calG_a)$, 
since no vertex in $X$ is adjacent to $t$ in the resulting graph.
For the other direction, let $(G_1 \in \calG_1, \ldots, G_a \in \calG_a)$
be a solution to the \PDualCMCshort instance.
Consider some $x \in V(H)$ that is incident with two edges $e_{i_y} = xy$ and $e_{i_z} = xz$ in $H$.
Let $G_{i_y} = G(i_y, v_x, v_y)$ and $G_{i_z} = G(i_z, w_x, w_z)$.
In $G_{i_y} \cup G_{i_z}$, 
if (i) $v_x = w_x$, then there is a component $C$ not containing $t$, but containing
$v_x=w_x$, $v_y$, and $w_z$, while if (ii) $v_x \neq w_x$, then $t$ is connected to all
vertices of $V_x$, and hence also all vertices of $V_y$ and $V_z$.
Now suppose for a contradiction that the edges in $K$ corresponding to the graphs $\calG_i$ do not give a solution to $(H, K)$.
Then, for at least one vertex $x \in V(H)$, case (ii) has to apply, which allows us to conclude that 
$\bigcup_{i \in [a]} G_i$ was in fact connected as long as $K$ is connected (which can be easily ensured), a contradiction.

The main technical contribution of the following theorem is a fairly non-trivial construction 
that emulates the requirement that we have to choose one graph from each part of the partition $(\calG_1, \ldots, \calG_a)$.
To do so, we equip each graph $G_i \in \calG_i$ with a set of \emph{padding edges} with the following properties.
\begin{enumerate}
	\item For each $i \in [a]$ and distinct $G_{i_1}, G_{i_2} \in \calG_i$, the padding edges in $G_{i_1} \cup G_{i_2}$ connect the entire vertex set~$W$.
	\item For all $G_1 \in \calG_1, \ldots, G_a \in \calG_a$, the padding edges in $\bigcup_{i \in [a]} G_i$ leave $W$ (highly) disconnected.
\end{enumerate}
To make this construction work, we have to add many more vertices to $W$, and 
an additional challenge is to keep the size of $W$ sufficiently small
to obtain the desired lower bound. 
In particular, to obtain a lower bound of the form $(np)^{o(\log n/\mathrm{polyloglog}(n))}$ 
for \CMC, we require the dependence of the size of $W$ on $a$, 
the number of edges in the pattern graph,
to be no more than $2^{\Oh(a \mathrm{polylog}(a))}$.
\begin{theorem}\label{thm:main:reduction}
	There is a reduction that given a \PSI instance $(H, K)$, where
	$H$ is connected, and with
	$\card{V(H)} = h$, %
	$\card{E(H)} = a$,
	and $n$ being the number of vertices in each part of the partition of $V(K)$,
	constructs an equivalent \DualCMC instance $(W, G_1, \ldots, G_p, a)$ with
	$\card{W} = \Oh(h(4a)^a n)$, and $p = \card{E(K)}$
	in time $\card{W}^{\Oh(1)}$.
\end{theorem}
\begin{proof}
	Let $(H, K)$, $n$, $h$, and $a$ be as in the statement of the theorem.
	We show how to construct an instance $(W, G_1, \ldots, G_p, a)$ of \DualCMC 
	with the claimed properties.
	Throughout the proof,
	for $xy \in E(H)$, we let $E(V_x, V_y) = \{v_xv_y \in E(K) \mid v_x \in V_x, v_y \in V_y\}$.
	
	\newcommand\aprime{\rho}
	We choose a prime $\aprime$ with $\ceil{n^{1/a}} < \aprime \le 2\ceil{n^{1/a}}$ 
	(whose existence can be shown by elementary number theory) 
	and let $b = 2a$.
	Then we have that
	\begin{align}
		\label{eq:prime:p}
		(\aprime-1)^a \ge n \mbox{ and } \aprime^b \ge n^2.
	\end{align}
	We let $\bF_\aprime$ be the field of integers modulo $\aprime$ 
	and let $\bFstar_\aprime$ be its multiplicative subgroup, 
	i.e., $\bFstar_\aprime = \{1, \ldots, \aprime-1\}$.
	For each $x \in V(H)$, we fix an arbitrary injective map $f_x \colon V_x \to (\bFstar_\aprime)^a$.
	Note that the existence of such maps is guaranteed by \eqref{eq:prime:p}.
	For each $xy \in E(H)$, let $g_{xy} \colon V_x \times V_y \to (\bFstar_\aprime)^b$
	be the map defined as $g_{xy}(v_xv_y) = f_x(v_x) \circ f_y(v_y)$.
	
	For each $x \in V(H)$, 
	we let $W_x = (\bF_\aprime \times \{0,\ldots,b\})^a$
	and $\hat{W_x} = (\bF_\aprime \times \{0\})^a \subseteq W_x$,
	and we embed $f_x$ into $\hat{W_x}$ with the following map 
	$\hat{f_x} \colon V_x \to \hat{W_x}$.
	For each $v_x \in V_x$, we let 
	\[\hat{f_x}(v_x) = ((f_x(v_x)[1], 0), \ldots, (f_x(v_x)[a], 0)).\]
	The vertex set of the \DualCMC instance is $W = \{t\} \cup \bigcup_{x \in V(H)} W_x$,
	where $t$ is a new vertex.
	
	\begin{figure}
		\centering
		\includegraphics[width=.8\textwidth]{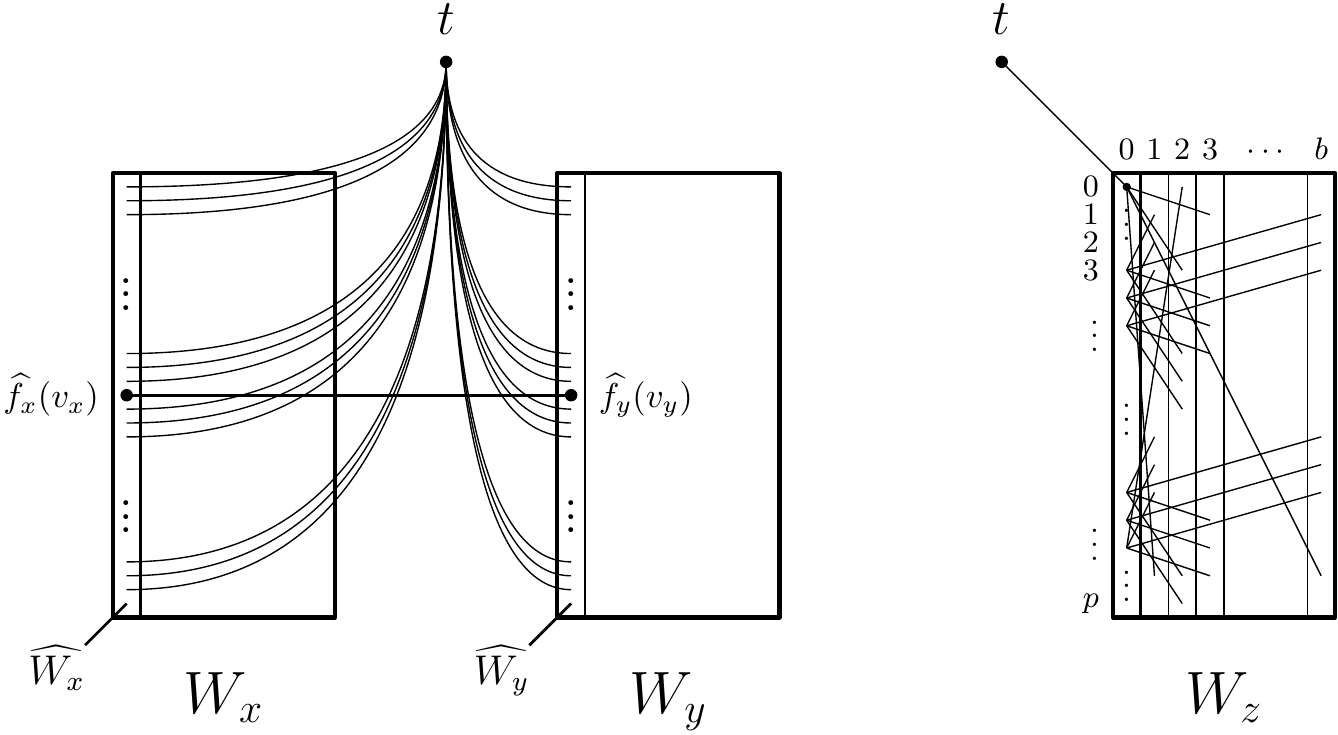}
		\caption{On the left an illustration of a set $A(\alpha, v_x, v_y)$ and on the right 
			an illustration of a padding set $\pad(\alpha, v_x, v_y, z)$.
			Note that the vector $g_{xy}(v_x, v_y)$ determines the shape of the stars 
			in $\pad(\alpha, v_x, v_y, z)$. 
			We would like to remark though that the picture presented here is 
			drastically simplified and should not be interpreted ``literally'', 
			as most dimensions of the elements of the sets considered here do not show.}
		\label{fig:red:1}
	\end{figure}
	
	The graphs of the \DualCMC instance are constructed as follows, see \cref{fig:red:1} for an illustration.
	Throughout the following, let $E(H) = \{e_1, \ldots, e_a\}$.
	Let $\alpha \in [a]$ and $e_\alpha = xy$;
	for each $v_xv_y \in E(V_x, V_y)$,
	we let
	\begin{align}
	    \label{eq:red:Aedges}
		A(\alpha, v_x, v_y) = \{\hat{f_x}(v_x), \hat{f_y}(v_y)\} \cup 
			\left\lbrace
				tz \mid z \in (\hat{W_x} \cup \hat{W_y}) \setminus \{\hat{f_x}(v_x), \hat{f_y}(v_y)\}.
			\right\rbrace
	\end{align}
	We obtain the graph $G(\alpha, v_x, v_y)$ by taking $A(\alpha, v_x, v_y)$ 
	and adding the following set of ``padding'' edges $\pad(\alpha, v_x, v_y, z)$ for each $z \in V(H)$,
	containing the following edges between vertices of $W_z \cup \{t\}$:
	\begin{align}
		\nonumber
		&\forall q_1, \ldots, q_{\alpha-1}, q_{\alpha+1}, \ldots, q_a \in \bF_\aprime \times \{0, \ldots, b\} \colon 
		\{t, (q_1, \ldots, q_{\alpha-1}, (0,0), q_{\alpha+1}, \ldots, q_a)\}, \mbox{ and } \\
		\nonumber
		&\forall q_1, \ldots, q_{\alpha-1}, q_{\alpha+1}, \ldots, q_a \in \bF_\aprime \times \{0, \ldots, b\} 
		~~~\forall r \in \bF_\aprime
		~~~\forall i \in [b] \colon \\
		\label{eq:red:star}
		&~~~~~\{(q_1, \ldots, q_{\alpha-1}, (r, 0), q_{\alpha+1}, \ldots, q_a),
			(q_1, \ldots, q_{\alpha-1}, (r + g_{xy}(v_x, v_y)[i], i), q_{\alpha+1}, \ldots, q_a)\}
	\end{align}
	Note that the edges described in \eqref{eq:red:star} form a star whose center is in $\hat{W_z}$.
	In particular, for each $r \in \bF_\aprime$, we construct one star whose center is $(r, 0)$;
	the shape of each such star is determined by the vector $g_{xy}(v_x, v_y)$.
	The (edge set of the) graph $G(\alpha, v_x, v_y)$ is then obtained as:
	\[
		G(\alpha, v_x, v_y) = A(\alpha, v_x, v_y) \cup \bigcup\nolimits_{z \in V(H)} \pad(\alpha, v_x, v_y, z)
	\]
	We make an easy but crucial observation.
	\begin{observation}\label{obs:red:span}
		Let $\alpha \in [a]$, $e_\alpha = xy$, and $v_xv_y \in E(V_x, V_y)$.
		Then, $G(\alpha, v_x, v_y)$ spans $W$.
	\end{observation}
	
	The graphs $G(\alpha, v_x, v_y)$ make up the graphs of the \DualCMC instance;
	note we have $\card{E(K)}$ of them
	and therefore $m = \card{E(K)}$ as required by the statement of the theorem.
	Moreover, since $b = 2a$ and $\aprime \le 2\ceil{n^{1/a}}$, we have that
	\begin{align*}
		\card{W} = 1 + h(\aprime(b+1))^a = 1 + h \aprime^a (2a+1)^a 
			= \Oh(h (2n^{1/a})^a (2a)^a)
			= \Oh(h 2^a(2a)^a n) 
			= \Oh(h (4a)^a n),
	\end{align*}
	as claimed.
	It remains to prove the correctness of the reduction.
	
	\begin{claim}\label{claim:red:padding:works}
		Let $\alpha \in [a]$, $e_\alpha = xy$, and $v_xv_y, v_x'v_y' \in E(V_x, V_y)$ be distinct.
		Then, 
		\(G(\alpha, v_x, v_y) \cup G(\alpha, v_x', v_y')\) 
		is connected.
	\end{claim}
	\begin{claimproof}
		We show that for each vertex $w \in W \setminus \{t\}$, 
		there is a path from $w$ to $t$ using only the edges 
		from $G(\alpha, v_x, v_y) \cup G(\alpha, v_x', v_y')$.
		To do that, let $z \in V(H)$ be such that $w \in W_z$, 
		and let $q_1, \ldots, q_a \in \bF_\aprime \times \{0,\ldots,b\}$ 
		be such that $w = (q_1, \ldots, q_a)$.
		We will be able to ``move within'' the $\alpha$-th coordinate of $(q_1, \ldots, q_a)$,
		and for better readability we introduce the following shorthand: 
		for $(r, \beta) \in \bF_\aprime \times \{0,\ldots,b\}$, we let
		$(r, \beta)_\alpha$ stand for $(q_1, \ldots, q_{\alpha-1}, (r, \beta), q_{\alpha+1}, \ldots, q_a)$.
		Let $q_\alpha = (r, \beta)$.
		We will show that we can reach $(0,0)_\alpha$, which is adjacent to $t$, from $w = (r, \beta)_\alpha$.
		
		Let $g_{xy}(v_xv_y) = (r_1, \ldots, r_b)$, $r_0 = 0$,  $g_{xy}(v_x'v_y') = (r_1', \ldots, r_b')$, and $r_0' = 0$.
		Since $v_xv_y$ and $v_x'v_y'$ are distinct, there is some $\gamma \in [b]$ such that $r_\gamma \neq r_\gamma'$.
		Let $r^* = r - r_\beta$ 
		and note that in the case $\beta \neq 0$, 
		there is an edge between $(r, \beta)_\alpha$ and $(r^*, 0)_\alpha$ in $G(a, v_x, v_y)$.
		Moreover, the following is a path in $G(\alpha, v_x, v_y) \cup G(\alpha, v_x', v_y')$: 
		\begin{align}
				\label{eq:red:path}
				(r^*, 0)_\alpha, (r^* + r_\gamma, \gamma)_\alpha, (r^* + r_\gamma - r_\gamma', 0)_\alpha, 
					(r^* + 2r_\gamma -r_\gamma', \gamma)_\alpha, (r^* + 2(r_\gamma-r_\gamma'), 0)_\alpha, \ldots
		\end{align}
		Since $\aprime$ is prime and $r_\gamma \neq r_\gamma'$, the path shown in~\eqref{eq:red:path} eventually leads to $(0, 0)_\alpha$,
		and therefore to $t$.
	\end{claimproof}
	
    \begin{claim}\label{claim:red:disconn:to:H:subgraph}
        For each $\alpha \in [a]$ with $e_\alpha = x_\alpha y_\alpha$, let $v_{x_\alpha}v_{y_\alpha} \in E(V_{x_\alpha}, V_{y_\alpha})$.
        If $\calG = \bigcup_{\alpha \in [a]} G(\alpha, v_{x_\alpha}, v_{y_\alpha})$ is disconnected,
        then $(H, K)$ is a \yes-instance.
    \end{claim}
    \begin{claimproof}
        To prove the claim, we need to show that $U = \{v_{x_\alpha} \mid \alpha \in [a]\}$ contains exactly one vertex from each $V_x$ where $x \in V(H)$.
        The existence of the necessary edges to show the $H$-subgraph in $K$ is then guaranteed by the choice made in the statement of the claim.
        To do so, we show that all vertices in 
        $\hat{U} = \{\hat{f_{x_\alpha}}(v_{x_\alpha}) \mid \alpha \in [a]\}$
        are in the same connected component of $\calG$, and in particular that this component does not contain $t$.
        This yields the desired property, for if there was some $x \in V(H)$ such that $\card{U \cap V_x} \ge 2$, 
        then for $\hat{W_x}$, $\card{\hat{U} \cap \hat{W_x}} \ge 2$, 
        which implies that for some distinct $\alpha, \alpha' \in [a]$,
        and distinct $v_x, v_x' \in V_x$, 
        $\calG$ contains both $G(\alpha, v_x, v_y)$ and $G(\alpha', v_x', v_{y'})$ for some $y \neq y'$ and $v_y \in V_y$ and $v_{y'} \in V_{y'}$.
        As $\hat{f_x}(v_x)$ is adjacent to $t$ in $G(\alpha', v_x', v_{y'})$ since $v_x \neq v_x'$, we obtain a contradiction.
        
        Since $\calG$ is disconnected, there is some $z \in V(H)$ such that $W_z$ contains a vertex 
        that is in a different connected component than $t$; call this vertex $\w_0$.
		We show that $\w_0$ has a path to a vertex in $\hat{W_z}$ in $\calG$.
		For all $\alpha \in [a]$, we let $g_{x_\alpha y_\alpha}(v_{x_\alpha}, v_{y_\alpha}) = (r_1^\alpha, \ldots, r_b^\alpha)$ and $r_0^\alpha = 0$.
        Let $\w_0 = ((s_1, \gamma_1), \ldots, (s_a, \gamma_a)) \in \bF_\aprime \times \{0, \ldots, b\}$.
        For all $\alpha \in [a]$, 
        let 
        \[
            \w_\alpha = ((s_1 - r^1_{\gamma_1}, 0), \ldots, (s_\alpha - r^\alpha_{\gamma_\alpha}, 0), 
                (s_{\alpha+1}, \gamma_{\alpha+1}), \ldots, (s_a, \gamma_a)).
        \]
        Then, $\w_0, \w_1, \ldots, \w_a$ is a path in $\calG$, since for all $\alpha \in [a]$, the edge $\w_{\alpha-1} \w_\alpha$ is in 
        $G(\alpha, v_{x_\alpha}, v_{y_\alpha})$.
        Moreover, $\w_a \in \hat{W_z}$. 
        For each $\alpha$ such that $z$ is an endpoint of $e_\alpha = x_\alpha y_\alpha$, assume $z = x_\alpha$,
        only one vertex in $\hat{W_z}$ is not adjacent to $t$ 
        in $G(\alpha, x_\alpha, y_\alpha)$;
        namely $\hat{f_{x_\alpha}}(v_{x_\alpha})$, see \eqref{eq:red:Aedges}.
        Therefore, $\w_a = \hat{f_{x_\alpha}}(v_{x_\alpha})$ 
        for \emph{every} such $\alpha$.
        This allows us to furthermore conclude that $\hat{f_{y_\alpha}}(v_{y_\alpha})$, which is connected to $\hat{f_{x_\alpha}}(v_{x_\alpha})$ by an edge in $G(\alpha, v_{x_\alpha}, v_{y_\alpha})$, see \eqref{eq:red:Aedges}, is in the same connected component as $\w_0$, and therefore disconnected from $t$.
        Since $H$ is connected, we can iteratively apply this argument to conclude that all vertices in 
        $\hat{U} = \{\hat{f_{x_\alpha}}(v_{x_\alpha}) \mid \alpha \in [a]\}$ 
        are in the same connected component of $\calG$ that does not contain $t$, 
        which concludes the proof by the above discussion.
    \end{claimproof}
    
    \begin{claim}\label{claim:red:H:subgraph:to:disconn}
        If $\{v_x \mid x \in V(H)\}$ is an $H$-subgraph in $K$, 
        then $\hat{U} = \{\hat{f_x}(v_x) \mid x \in V(H)\}$ are disconnected from $t$
        in $\calG = \bigcup_{\alpha \in [a], e_\alpha = x_\alpha y_\alpha} G(\alpha, v_{x_\alpha}, v_{y_\alpha})$.
    \end{claim}
    \begin{claimproof}
		Suppose for a contradiction that there is some vertex in~$\hat{U}$ 
		that is connected by a path to~$t$.
		In particular, let $x \in V(H)$ be such that there is such a vertex $\hat{f_x}(v_x) = \hat{u}$ 
		that has a path $P$ to $t$ in $\calG$ that is fully contained in $W_x \cup \{t\}$.
		Let $\w$ be the other endpoint of the edge in $P$ that is incident with $t$.
		By construction we have that $\w \neq \hat{u}$.
		Let $\w = ((r_1, 0), \ldots, (r_a, 0))$ and $\hat{u} = ((r_1^*, 0), \ldots, (r^*_a, 0))$.
		Since $\w \neq \hat{u}$, there is some $\alpha \in [a]$ such that 
		$r_\alpha \neq r^*_\alpha$.
    	Since there is no edge $v_{x_\alpha}'v_{y_\alpha}' \in E(V_x, V_y)$ with $v_{x_\alpha}v_{y_\alpha} \neq v_{x_\alpha}'v_{y_\alpha}'$ 
    	such that $G(\alpha, v_{x_\alpha}', v_{y_\alpha}') \subseteq \calG$, we conclude that all vertices reachable from $\w$ in $\calG$ have the value $(r_\alpha + g_{x_\alpha y_\alpha}(v_{x_\alpha}v_{y_\alpha})[\beta], \beta)$
    	for some $\beta \in \{0,\ldots,b\}$ in the $\alpha$-th coordinate,
    	which means in particular for $\beta = 0$ that their $\alpha$-th coordinate is $(r_\alpha, 0) \neq (r_\alpha^*, 0)$, and therefore such a path $P$ cannot exist, a contradiction.
    \end{claimproof}
    
    To conclude the correctness proof, suppose that $(H, K)$ is a \yes-instance. Then \cref{claim:red:H:subgraph:to:disconn} implies that $(W, G_1, \ldots, G_p, a)$ is a \yes-instance as well.
    Conversely, suppose that $(W, G_1, \ldots, G_p, a)$ is a \yes-instance.
    By \cref{claim:red:padding:works}, we know that a solution to the instance must be of the form as required by \cref{claim:red:disconn:to:H:subgraph} which yields that $(H, K)$ is a \yes-instance.
\end{proof}

The following theorem due to Marx~\cite{Marx2010} 
has an immediate consequence for the parameterized complexity of \DualCMC.

\begin{theorem}[Marx~\cite{Marx2010}]\label{thm:marx}
	\PSI where the host graph has $n$ vertices and the pattern graph $h$ edges
	parameterized by $h$ is \Wone-hard,
	and cannot be solved in time $f(h)n^{o(h/\log h)}$, for any computable function $f$,
	unless the \ETH fails.
\end{theorem}

Note that we can always assume that $H$ is connected in the previous theorem. 
If not, then we add a new color class to $H$ containing a single universal vertex $x_H$,
as well as a universal vertex $x_K$ to $K$ with $V_{x_H} = \{x_K\}$.
Therefore, \cref{thm:main:reduction,thm:marx} imply:
\thmwhard*

To obtain the quasi-polynomial time lower bound for \CMC under the \ETH, 
we need some more ingredients proved in \cref{sec:fine:Marx};
and we wrap it up in \cref{sec:wrap}.

%% file: cybt.tex
\section{Fine-tuned reduction of Marx}
\label{sec:fine:Marx}

Two connected subgraphs $H_1$ and $H_2$ of a graph $H$ \emph{touch}
if they share a vertex or if there is an edge of $H$ with one endpoint in $V(H_1)$
and another endpoint in $V(H_2)$. 
An \emph{embedding} of a graph $G$ in a graph $H$ is an assignment $\phi$
that assigns to every $v \in V(G)$ a nonempty connected subgraph $\phi(v)$ in $H$
(called the \emph{branch set of $v$})
such that for every edge $uv \in E(G)$, the subgraphs $\phi(u)$ and $\phi(v)$ touch. 
The \emph{depth} of an embedding $\phi$
is $\max_{x \in V(H)} |\{v \in V(G)~|~x \in \phi(v)\}|$, that is, the maximum 
number of subgraphs $\phi(v)$ that meet in a single vertex of $H$. 
In literature, such embeddings are sometimes also called \emph{congested minor models}
and depth is called \emph{congestion} or \emph{ply}. 

In this section, we prove the following version of Theorem~3.1 of~\cite{Marx2010}. 
This version is more fine-tuned for our application: a lower bound for binary CSPs parameterized by the number of constraints. 

\begin{theorem}\label{thm:embedding}
There exists a constant $C$ and a polynomial-time algorithm that, given an integer $k \geq 2$ and a graph $G$ with $n$ vertices and $m$ edges,
outputs a graph $H$ with $|V(H)|+|E(H)| \leq k$ and an embedding of $G$ into $H$. 
Furthermore, with probability at least $0.5$, the depth of the output embedding is at most 
\[ C\left(1+k^{-1}(n+m)\right) \cdot \log k.\]
\end{theorem}

We start with a few simplification steps in the proof of Theorem~\ref{thm:embedding}.
\begin{enumerate}
\item For $k = \Oh(1)$, the algorithm of Theorem~\ref{thm:embedding} can output just a single-vertex
graph $H$ and embed the whole $G$ into that vertex. Thus, we can assume $k$ is larger
than any fixed constant, in particular, $k \geq 8$. 
\item Our constructed graph $H$ will have maximum degree at most $3$ and will 
satisfy $|V(H)| \geq \lfloor k/4 \rfloor \geq k/8$, so any isolated vertices of $G$ can be equidistributed across the vertices of $H$
adding at most $1+k^{-1}n$ to the congestion. Thus, we can assume $G$ has no isolated vertices.\label{p:requirek4}
\item We can assume $G$ has maximum degree at most $3$: for every vertex $v$ of degree $\mathrm{deg}_G(v) > 3$, we replace $v$ with a cycle on $\mathrm{deg}_G(v)$ vertices
and equidistribute the edges incident to $v$ among the vertices of the cycle. This increases $n+m$ by at most a factor of $3$.
Since the original graph $G$ is a minor of the modified one, it suffices to find an embedding of the modified graph.
\item Finally, if at this point $n + m \leq k$, then we can output $H=G$ (and, if $|V(G)| < \lfloor k/4 \rfloor$, a number of isolated vertices to have $\lfloor k/4 \rfloor$ vertices in $G$
   as it is required in Point~\ref{p:requirek4} above)
and a trivial embedding, so assume $n + m > k$. 
\end{enumerate}

In what follows we will need the following variant of Chernoff bound (this is the same as used in~\cite{Marx2010}). 
\begin{theorem}[Chernoff bound]\label{thm:Chernoff}
Let $X_1,\ldots,X_n$ be independent random variables with values in $\{0,1\}$.
Let $X = \sum_{i=1}^n X_i$, let $\mu = \mathbb{E} X$, and let $r \geq \mu$ be a real.
Then, 
  \[ \mathrm{Pr} \left(X > \mu + r\right) \leq \left(\frac{e \mu}{r}\right)^r. \]
\end{theorem}
\begin{proof}
The case $\mu = 0$ is trivial, so assume $\mu > 0$.
Let $\delta = r/\mu \geq 1$. 
We use the standard Chernoff bound as in Wikipedia:
\[ \mathrm{Pr}\left(X > (1+\delta)\mu\right) \leq \left(\frac{e^\delta}{(1+\delta)^{1+\delta}}\right)^\mu. \]
Hence
\begin{align*}
\mathrm{Pr}\left(X > \mu+r\right) &= \mathrm{Pr}\left(X > (1+\delta)\mu \right) \leq \left(\frac{e^\delta}{(1+\delta)^{1+\delta}}\right)^\mu \\
  & \leq \left(\frac{e}{1+\delta}\right)^{\mu(1+\delta)} = \left(\frac{e\mu}{\mu+r}\right)^{\mu+r} \leq \left(\frac{e\mu}{r}\right)^r.
\end{align*}
\end{proof}

\subsection{Concurrent flows and expanders}

For a constant $\alpha > 0$, we say that a (multi)graph $H$ is an \emph{$\alpha$-expander}
if for every nonempty set $S \subseteq V(G)$ of size at most $|V(G)|/2$, the number
of edges with exactly one endpoint in $S$ is at least $\alpha |S|$. 
For the choice of the graph $H$ in Theorem~\ref{thm:embedding}, we can rely on any polynomial-time construction of a constant-degree expander.
We encapsulate it in the following standard claim.

\begin{theorem}\label{thm:expanders}
There exists a constant $\delta>0$ and polynomial-time algorithm that, given an integer $\ell \geq 1$,
outputs a (simple) graph $H_\ell$ on $\ell$ vertices that is a $\delta$-expander
and has maximum degree at most $3$.
\end{theorem}

\paragraph{Obtaining a multicommodity flow.}
We now recall some tools from~\cite{GroheM09,Marx2010}. 
Fix a graph $G$ and a set $W \subseteq V(G)$.
A pair $(A,B)$ is a \emph{separation} in $G$ if $A \cup B = V(G)$ and there is no edge
between $A \setminus B$ and $B \setminus A$. The \emph{sparsity} of $(A,B)$ (w.r.t.~$W$) 
  is defined as 
\[ \alpha^W(A,B) = \frac{|A \cap B|}{|A \cap W| \cdot |B \cap W|}. \]
Let $\alpha^W(G)$ be the minimum sparsity of a separation of $G$. 
We have the following observation.
\begin{lemma}\label{lem:alpha-delta}
If $G$ is a $\delta$-expander of maximum degree $3$, then 
\[\alpha^{V(G)}(G) \geq \frac{\delta}{3+\delta} |V(G)|^{-1}.\] 
\end{lemma}
\begin{proof}
Let $(A,B)$ be a separation of $G$. Without loss of generality assume that $|A| \leq |B|$,
    so $|A \setminus B| \leq |V(G)|/2$
Then, by the expander property for the set $A \setminus B$ we have
\[ \delta|A \setminus B| \leq 3|A \cap B|. \]
Hence
\[ |A| \leq \frac{3+\delta}{\delta} |A \cap B|. \]
Therefore
\[ \alpha^{V(G)}(A,B) = \frac{|A \cap B|}{|A| \cdot |B|} \geq \frac{\delta}{3+\delta} \cdot |V(G)|^{-1}. \]
\end{proof}

For the definition of a concurrent flow and sampling from it, we follow the notation
of~\cite{HKPS20}.

A \emph{flow} $f_{u, v}$ between two vertices $u$ and $v$ is a weighted collection of pairwise distinct paths between $u$ and $v$ (if $u = v$ then the only possible path in $f_{u, v}$ has length 0).
The \emph{units of flow sent through a vertex $w$} by $f_{u, v}$ is the sum of the weights of the paths in $f_{u, v}$ that contain~$w$.
The \emph{value} of $f_{u, v}$ is defined as the units of flow sent through $u$ (or equivalently through $v$).
A \emph{concurrent flow of value $\nu$
	and congestion $\gamma$} is a collection of $|V(G)|^2$ flows $(f_{u,v})_{(u,v) \in V(G) \times V(G)}$
such that
\begin{itemize}
	\item $f_{u,v}$ sends exactly $\nu$ units of flow from $u$ to $v$; and
	\item for each vertex $w$ the total flow over all flows $f_{u, v}$ sent through $w$ is at most $\gamma$.
\end{itemize}

We need the following relation between the minimum possible congestion and 
sparsity of separations.
\begin{theorem}[\cite{FeigeHL08,LeightonR99}]\label{thm:fhl}
Let $\gamma$ be the minimum possible congestion of concurrent flow of value $1$ in a graph $G$.
Then there exists a separation with sparsity (w.r.t.~$|V(G)|$) $\Oh(\log |V(G)| / \gamma)$. 
\end{theorem}

From the above toolbox we obtain the following corollary.
\begin{lemma}\label{lem:getH}
There exists a constaint $c$ and an algorithm that, given an integer $\ell \geq 2$, in time polynomial in $\ell$
computes a graph $H$ with $\ell$ vertices and maximum degree at most $3$
and a concurrent flow $(f_{u,v})_{u,v \in V(H)}$ of congestion at most $c \ell \log \ell$. 
\end{lemma}
\begin{proof}
Apply Theorem~\ref{thm:expanders} to $\ell$, obtaining a $\delta$-expander $H := H_\ell$.
As $\alpha^{V(H)}(H) \geq \frac{\delta}{3+\delta} \ell^{-1}$ by Lemma~\ref{lem:alpha-delta},
there exists in $H$ a concurrent flow of value $1$ and congestion at most $c \dot \ell \log \ell$
for some universal constant $c$ (that depends on $\delta$ and the constants hidden in Theorem~\ref{thm:fhl}). 
Finally, we observe that the problem of finding a concurrent flow of value $1$ and minimum possible congestion can be formulated as a linear program, and therefore solved in polynomial time. 
\end{proof}
 
\subsection{Constructing the embedding}

We set $\ell := \lfloor k/4 \rfloor$; note that $\ell \geq 2$ due to the assumption $k \geq 8$. 
We apply Lemma~\ref{lem:getH} to $\ell$ obtaining a graph $H$ on $\ell$ vertices and maximum degree at most $3$, and a concurrent flow $(f_{u,v})_{u,v \in V(H)}$ of value $1$ and congestion
at most $c \ell \log \ell \leq ck \log k$. Without loss of generality, we assume $c \geq 1$.
Note that $|V(H)|+|E(H)| \leq \ell+3\ell \leq k$ as desired.
We assume $V(H) = [\ell]$ for ease of the notation.

Similarly as in~\cite{GroheM09,HKPS20,Marx2010}, we treat $f_{u,v}$ as a probability distribution over paths from $u$ to $v$:
the probability of choosing a path equals the amount of flow passed along this path (its weight in $f_{u,v}$). 

The crucial tool for the construction is the following analysis.
\begin{claim}\label{cl:make-paths}
Fix integer $p \geq 1$ and consider the following random process.
For every $x \in V(H)$ and every $i \in [p]$, randomly pick (with uniform distribution) a vertex $Y(x,i) \in V(H)$ and then a path $P(x,i)$
from $x$ to $Y(x,i)$ according to the distribution $f_{x,Y(x,i)}$. 
Then, with probability more than $0.9$, the family of paths $\mathcal{P} = \{P(x,i)~|~x \in V(H), i \in [p]\}$ has congestion at most $10cp \log \ell$ (where $c$ is the constant from Lemma~\ref{lem:getH}).
\end{claim}
\begin{proof}
Fix a vertex $w \in V(H)$. 
For $x,y \in V(H)$, let $f_{x,y}(w)$ be the amount of flow $f_{x,y}$ that passes through $w$.
For $x \in V(H)$ and $i \in [p]$, let $X(x,i)$ be a $\{0,1\}$-valued random variable indicating if $w$ lies on the path $P(x,i)$.
Then,
\[ \mathbb{E} X(x,i) = \mathrm{Pr} \left(X(x,i) = 1\right) = \ell^{-1} \sum_{y \in V(H)} \mathrm{Pr} \left(w \in P(x,i)~|~Y(x,i) = y \right) = \ell^{-1} \sum_{y \in V(H)} f_{x,y}(w). \]
Hence, if we define $X = \sum_{x \in V(H)} \sum_{i=1}^p X(x,i)$, we have 
\[ \mathbb{E} X = \sum_{x \in V(H)} \sum_{i=1}^p \ell^{-1} \sum_{y \in V(H)} f_{x,y}(w) = p\ell^{-1} \sum_{x,y \in V(H)} f_{x,y}(w) \leq p\ell^{-1} c\ell\log\ell = cp\log \ell. \]
By the aforementioned version of the Chernoff bound (Theorem~\ref{thm:Chernoff}), we obtain:
\[ \mathrm{Pr}\left(X > 10cp \log \ell \right) \leq \left(\frac{e \mathbb{E} X}{9cp \log \ell}\right)^{9cp \log \ell} \leq e^{-9cp \log \ell} \leq \ell^{-9} < \ell^{-1}/10.\]
In the last two inequalities we used $c \geq 1$ and $\ell \geq 2$. 
Thus, by the union bound, the probability that for every $w \in V(H)$ at most $10cp\log \ell$ of $\mathcal{P}$ pass through $w$ is at least $0.9$.
\end{proof}

We are now ready to construct the final embedding $\phi$ of $G$ into $H$. 
We start with spliting $V(G)$ into $\ell$ buckets of as equal as possible sizes: that is, 
we take arbitrary $\zeta : V(H) \to [\ell]$ such that for every $a \in [\ell]$, 
we have $|\zeta^{-1}(i)| \in \{ \lfloor |V(G)|/\ell \rfloor, \lceil |V(G)|/\ell \rceil \}$. 
We start by setting $\phi(v) = \{\zeta(v)\}$ for every $v \in V(G)$.
Then, for every $xy \in E(G)$ with $\zeta(x) \neq \zeta(y)$ we proceed as follows:
\begin{enumerate}
\item Uniformly at random pick $Z(xy) \in V(H)$.
\item Sample a path $P(xy,x)$ from $x$ to $Z(xy)$ according to the distribution $f_{x,Z(xy)}$ and add $P(xy,x)$ to $\phi(x)$.
\item Sample a path $P(xy,y)$ from $y$ to $Z(xy)$ according to the distribution $f_{y,Z(xy)}$ and add $P(xy,y)$ to $\phi(y)$.
\end{enumerate}
Clearly, the constructed $\phi$ is indeed an embedding of $G$ to $H$. It remains to bound its congestion. 

To this end, consider a vertex $w \in V(H)$ and a vertex $v \in V(G)$ such that $w \in \phi(v)$.
We say that $w$ is a \emph{type-0} member of $\phi(v)$ if $\zeta(v) = w$.
We say that $w$ is a \emph{type-1} member of $\phi(v)$ if for some $xy \in E(G)$, with $\zeta(x) < \zeta(y)$, 
$w$ lies on the sampled path $P(xy,x)$. 
We say that $w$ is a \emph{type-2} member of $\phi(v)$ if for some $xy \in E(G)$, with $\zeta(x) < \zeta(y)$, 
$w$ lies on the sampled path $P(xy,y)$. 
Note that $w$ is a member of at least one type (but possibly of multiple types). 

A fixed vertex $w$ is a type-0 member of exactly $|\zeta^{-1}(w)| \leq 1+\ell^{-1}n$ sets $\phi(v)$.
The crucial observation is that the number sets $\phi(v)$ which $w$ is a type-1 member can be bounded using Claim~\ref{cl:make-paths} with $p := 3(1+\ell^{-1}n)$. 
Indeed, for a fixed vertex $a \in V(H)$, let $E_{a,1}$ be the set of edges $xy \in E(G)$ with $a = \zeta(x) < \zeta(y)$.
Since $G$ is of maximum degree at most $3$, we have $|E_{a,1}| \leq 3|\zeta^{-1}(a)| \leq 3(1+\ell^{-1}n) = p$.
Crucially, the sets $(E_{a,1})_{a \in V(H)}$ are pairwise disjoint. Thus, the choices of $Z(xy)$ for all $a \in V(H)$ and $xy \in E_{a,1}$ are independent.
Hence, by Claim~\ref{cl:make-paths}, with probability at least 0.9 every $w \in V(H)$ is a type-1 member of at most $10cp\log \ell$ sets $\phi(v)$.

A symmetrical argument shows that with probability at least 0.9 every $w \in V(H)$ is a type-2 member of at most $10cp \log \ell$ sets $\phi(v)$. 
Consequently, by the union bound, with probability at least 0.8 the congestion of the constructed embedding is at most
\[ (1+\ell^{-1}n) + 2 \cdot 10c (1+\ell^{-1}n) \log \ell \leq (20c+1)(1+\ell^{-1}n)\log \ell \leq (120c+6)(1+k^{-1}n)\log k. \]
In the last inequality we used $\ell = \lfloor k/4 \rfloor$ and $k \geq 8$. 
This finishes the proof of Theorem~\ref{thm:embedding}.

%% file: wrap-up.tex
\newcommand\ThreeSat{\textsc{$3$-Cnf-Sat}\xspace}
\newcommand\satinst{\psi}
\newcommand\fail{\ensuremath{\mathsf{fail}}\xspace}
\newcommand\true{\ensuremath{\mathsf{true}}\xspace}
\newcommand\false{\ensuremath{\mathsf{false}}\xspace}
\newcommand\calC{\mathcal{C}}
\newcommand\calD{\mathcal{D}}
\newcommand\calR{\mathcal{R}}
\newcommand\congestion{\mathsf{con}}
\newcommand\cmcinst{\mathcal{I}}
\section{Wrapping up the proof of quasi-polynomial time lower bound}\label{sec:wrap}
In this section we finish the proof of the quasi-polynomial time lower bound for \CMC under the \ETH.
The lower bound follows from an application of \cref{thm:main:reduction} to an instance of \PSI 
constructed from a \ThreeSat instance 
with the use of \cref{thm:embedding} with appropriate choices of parameters along the way.
\thmhard*
\begin{proof}
	We give a randomized reduction from \ThreeSat to \DualCMC that succeeds with probability at least $0.5$.
	
	Let $\satinst$ be the \ThreeSat instance on $N$ variables and $M$ clauses,
	and let $G$ be the incidence graph of $\satinst$.
	We call the vertices in $G$ corresponding to variables in $\satinst$ the \emph{variable vertices},
	and the vertices corresponding to clauses \emph{clause vertices}.
	Let $G$ be the incidence graph of $\satinst$ and let $k = \sqrt{N + M}$.
	We apply \cref{thm:embedding} on $(G, k)$.
	Let $C$ be its constant, let $H$ be the graph, and $\phi$ the embedding of $G$ in $H$ returned by its algorithm.
	If the congestion of this embedding is more than $C(1+ k^{-1}(N+M))\cdot \log k$, then we report \fail. 
	Note that this happens with probability less than $0.5$.
	We may therefore assume that the congestion $\congestion(\phi)$ of the embedding $\phi$ is at most 
	\begin{align}
	    \label{eq:wrap:congestion}
		\congestion(\phi) \le C(1+ k^{-1}(N+M))\cdot \log k = \Oh(\sqrt{N+M}\log\sqrt{N+M}),
	\end{align}
	and that 
	\begin{align}
	    \label{eq:wrap:H}
		\card{V(H)} + \card{E(H)} \le k = \sqrt{N + M}.
	\end{align} 
	
	We now give a CSP instance $\satinst_H$ 
	that will be transformed to a \PSI instance to which we can apply \cref{thm:main:reduction}.
	To obtain $\satinst_H$,
	first think of $\satinst$ as a CSP instance $\satinst_G$ on the graph $G$, in the following way:
	The variables of $\satinst_G$ are the vertices of $G$.
	The domain of each variable corresponding to a variable vertex is $\{\true,\false\}$,
	and the domain of each variable corresponding to a clause vertex it is $[\ell]$,
	where $\ell \le 3$ is the number of literals in that clause.
	Each truth assignment to the variables of $\satinst$ naturally corresponds to a valuation 
	of the variable vertex variables in $\satinst_G$.
	If such an assingment is a satifying assignment, then each clause vertex variable can receive 
	a value that indicates which literal satisfies it.
	To ensure that this only works when a truth assignment is indeed satisfying, we add the following constraints to $\satinst_G$:
	Let $vz$ be an edge in $G$, where $v$ is a variable vertex and $z$ is a clause vertex, corresponding to a clause of size $\ell$.
	Suppose $v$ is the $i$-th literal in $z$.
	Let $R = \{\true, \false\} \times [\ell] \setminus \{i\}$.
	If $v$ occurs positively in $z$, then we add a constraint $((v, z), R \cup \{(\true, i)\})$ to $\satinst_G$,
	and if $v$ occurs negated, we add $((v, z), R \cup \{(\false, i)\})$ to $\satinst_G$.
	In either case, we denote the corresponding constraint by $((v, z), R_{vz})$.
	Such constraints ensure that if $z$ receives value $i$, then in any satisfying valuation,
	the truth assignment of $v$ satisfies $z$.
	
	We now obtain $\satinst_H$, a CSP instance on $H$, from $\satinst_G$ by ``routing $\phi$ through $\satinst_G$''.
	Concretely, this means the following.
	For each $v \in V(G)$, let $D_v$ denote the domain of $v$ in $\satinst_G$.
	The domain of each $w \in V(H)$ in $\satinst_H$ is $\calD_w = \bigtimes_{v \in V(G), w \in V(\phi(v))} D_v$.
	This way, assigning a value to $w$ in $\satinst_H$ corresponds to assigning values 
	to all $v \in V(G)$ that are mapped to $w$ via $\phi$.
	We add two types of constraints. 
	The first one simply checks that if a vertex $v \in V(G)$ is mapped to several vertices in $H$, 
	then the valuation stays consistent. 
	The second one checks the edge constraints from $\satinst_G$ 
	at a point where the subgraphs of their endpoints in $H$ touch.
	We use the following notation. 
	For $w \in V(H)$, $(v_1, \ldots, v_r) \in V(G)^r$ with $w \in \bigcap_{j \in [r]} V(\phi(v_j))$, 
		and $(i_1, \ldots, i_r) \in D_{v_1} \times \ldots \times D_{v_r}$,
	we let $\calD_w[v_1 \gets i_1, \ldots, v_r \gets i_r]$ 
	denote the subset of $\calD_w$, where for all $j \in [r]$, value of $v_j$ is fixed to $i_j$.
	In the following, domains of variables in $\satinst_H$ may get shrunk.
	For ease of notation, however, we always denote the current domain of $w \in V(H)$ by $\calD_w$.
	We do the following:
	\begin{enumerate}
		\item (\textit{Consistency.}) For each $v \in V(G)$ and each edge $w_1 w_2 \in E(H)$ such that $\{w_1, w_2\} \subseteq V(\phi(v))$,
			we add a constraint $((w_1, w_2), \bigcup_{i \in D_v} \calD_{w_1}[v \gets i] \times \calD_{w_2}[v \gets i])$. (Note that since $\phi(v)$ is connected, this indeed ensures consistency.)
		\item (\textit{Touching.}) For each $uv \in E(G)$, by definition, $\phi(u)$ and $\phi(v)$ touch. 
			In this part, we ensure that $\satinst_H$ respects the constraints in $\satinst_G$ imposed by the edges in $G$.
			For each $uv \in E(G)$, we do the following.
		\begin{enumerate}
			\item (\textit{Vertex Touching.}) Let $w \in V(\phi(u)) \cap V(\phi(v))$; 
				we restrict $\calD_w$ to 
				$$\bigcup\nolimits_{(i, j) \in R_{uv}} \calD_w[u \gets i, v \gets j].$$
			\item (\textit{Edge Touching.}) Let $wz \in E(H)$ such that $w \in V(\phi(u))$ and $z \in V(\phi(v))$;
				we add a constraint $((w, z), \calR^{uv}_{wz})$ to $\satinst_H$, where
					$$\calR^{uv}_{wz} = \bigcup\nolimits_{(i,j) \in R_{uv}} \calD_w[u \gets i] \times \calD_z[v \gets j].$$
		\end{enumerate}
	\end{enumerate}
	
	It follows directly from this description that $\satinst$ is satisfiable if and only if $\satinst_H$ is.
	Moreover, the domain of each variable in $\satinst_H$ (vertex in $H$) is at most $3^{\congestion(\phi)}$.

	We now transform $\satinst_H$ into a \PSI instance $(H, K)$;
	$H$ serves as the pattern graph and $K$ is obtained as follows.
    For all $x \in V(H)$, we add a set of vertices $V_x = \calD_x$ to $K$, i.e.,
    $V_x$ is the domain of the variable $x$ in $\satinst_H$.
    This way, a choice of one vertex per $V_x$ naturally corresponds to 
    a valuation of the variables in $\satinst_H$.
    To ensure that such a choice gives an $H$-subgraph in $K$ precisely when 
    the corresponding valuation satisfies $\satinst_H$,
    we add the following edges to $H$:
    For each $wz \in E(H)$,
    and each pair $u \in V_w$, $v \in V_z$,
    we add the edge $uv$ to $K$
    if and only if for each constraint of the form $((w, z), \calR)$ in $\satinst_H$,
    $(u, v) \in \calR$, in other words, it is satisfied by assigning~$u$ to~$w$ and~$v$ to~$z$.
    Again it is clear from the description that the instance $(H, K)$ 
    is equivalent to $\satinst_H$, and therefore to $\satinst$.
    If $H$ is not connected, 
    then we add a new color class to $H$ containing a single universal vertex $x_H$,
	as well as a universal vertex $x_K$ to $K$ with $V_{x_H} = \{x_K\}$.
	
    We now apply the reduction of \cref{thm:main:reduction} to $(H, K)$ 
    and obtain an instance $\cmcinst$ of \DualCMC.
    Let $L = N + M$.
    To analyze the size of $\cmcinst$, 
    recall that by \eqref{eq:wrap:H}, $\card{V(H)} + \card{E(H)} \le \sqrt{L}$,
    and since each color class of $K$ has at most $3^{\congestion(\phi)} = 2^{\Oh(\sqrt{L} \log \sqrt{L})}$ vertices,
    we have that $\card{V(K)} = 2^{\Oh(\sqrt{L} \log \sqrt{L})}$ and $\card{E(K)} = 2^{\Oh(\sqrt{L} \log \sqrt{L})}$,
    by \eqref{eq:wrap:congestion}.
    Let $n$ be the number of vertices in $\cmcinst$ and let $p$ be the number of labels.
    By \cref{thm:main:reduction},
    \begin{align*}
        n = \Oh\left(\card{V(H)}(4\card{E(H)})^{\card{E(H)}}\card{V(K)}\right) = 2^{\Oh(\sqrt{L}\log \sqrt{L})}, \mbox{ and } 
        p = \card{E(K)} = 2^{\Oh(\sqrt{L}\log\sqrt{L})},
    \end{align*}
    and $\cmcinst$ can be constructed in $2^{\Oh(\sqrt{L}\log \sqrt{L})}$ time.
    To conclude, an algorithm for \DualCMC running in $(np)^{o(\log n/(\log\log n)^2)}$ would yield a \ThreeSat algorithm running in
    \newcommand\valn{2^{\sqrt{L}\log\sqrt{L}}}
    \newcommand\logvaln{\sqrt{L}\log\sqrt{L}}
    \newcommand\loglogvaln{\log\sqrt{L} + \log\log\sqrt{L}}
    \begin{align*}
        (np)^{o(\log n/(\log\log n)^2)} 
            = \left(\valn\right)^{o\left(\frac{\logvaln}{(\loglogvaln)^2}\right)}
            = 2^{o\left(L\frac{(\log\sqrt{L})^2}{(\loglogvaln)^2}\right)} %
            = 2^{o(N+M)}
    \end{align*}
    time, contradicting the \ETH.
\end{proof}